\documentclass[conference]{IEEEtran}
\IEEEoverridecommandlockouts
\usepackage{cite}
\usepackage{amsmath,amssymb,amsfonts}
\usepackage{algorithmic}
\usepackage{graphicx}
\usepackage{textcomp}
\usepackage{xcolor}
\usepackage{amsthm}
\newtheorem{theorem}{Theorem}[section]
\newtheorem{lemma}{Lemma}
\def\BibTeX{{\rm B\kern-.05em{\sc i\kern-.025em b}\kern-.08em
    T\kern-.1667em\lower.7ex\hbox{E}\kern-.125emX}}
\begin{document}

\title{Robust Precoding Design for Coarsely Quantized MU-MIMO Under Channel Uncertainties--V0
\thanks{Dr. Qiu's  work  is partially supported  by  N.S.F. of China under Grant No.61571296 and N.S.F. of US under Grant No. CNS-1247778, No.
CNS-1619250. Dr. Wen's  work  is partially supported  by  N.S.F. of China under Grant No.61871265.}
}

\author{\IEEEauthorblockN{Lei Chu}
\IEEEauthorblockA{\textit{Dept. Electrical Engineering} \\
\textit{Shanghai Jiao Tong University}\\
Shanghai, China \\
leochu@sjtu.edu.cn}
\and
\IEEEauthorblockN{Fei Wen}
\IEEEauthorblockA{\textit{Dept. Electronic Engineering} \\
\textit{Shanghai Jiao Tong University}\\
Shanghai, China \\
wenfei@sjtu.edu.cn}
\and
\IEEEauthorblockN{Robert Caiming Qiu}
\IEEEauthorblockA{\textit{Dept. Electrical and Computer Engineering} \\
\textit{Tennessee Technological University}\\
Cookeville, TN 38505 USA\\
rqiu@tntech.edu}
}

\maketitle

\begin{abstract}
Recently, multi-user multiple input multiple output (MU-MIMO) systems with low-resolution digital-to-analog converters (DACs) has received considerable attention, owing to the capability of dramatically reducing the hardware cost. Besides, it has been shown that the use of low-resolution DACs enable great reduction in power consumption while maintain the performance loss within acceptable margin, under the assumption of perfect knowledge of channel state information (CSI).
In this paper, we investigate the precoding problem for the coarsely quantized MU-MIMO system without such an assumption.
The channel uncertainties are modeled to be a random matrix with finite second-order statistics.
By leveraging a favorable relation between the multi-bit DACs outputs and the single-bit ones, we first reformulate the original complex precoding problem into a nonconvex binary optimization problem.
Then, using the S-procedure lemma, the nonconvex problem is recast into a tractable formulation with convex constraints and finally solved by the semidefinite relaxation (SDR) method.
Compared with existing representative methods, the proposed precoder is robust to various
channel uncertainties and is able to support a MU-MIMO system with higher-order modulations, e.g., 16QAM.
\end{abstract}

\begin{IEEEkeywords}
MU-MIMO, low-resolution DACs, robust precoding.
\end{IEEEkeywords}

\section{Introduction}

MU-MIMO is one of the key techniques in the future communication system, in which the base station (BS) is equipped with large number of antennas \cite{b1}. Scaling up the number of transmit antennas brings up numerous advantages: enhanced system throughput, improved radiated energy efficiency and simplified algorithms of signal processing \cite{b2}.  On the other hand, gigantic increase in energy consumption, a key challenge to the use of massive MIMO, is brought by the increasing antenna quantity.

Recently, coarsely quantized MU-MIMO has been extensively investigated and recognized as an effective way to decrease the energy cost \cite{b117}. The authors in \cite{b3} show that, compared with massive MIMO systems with ideal DACs, the sum rate loss
in 1-bit massive MU-MIMO systems can be compensated by disposing approximately $\pi^2/4$ times more antennas at the BS. Besides, the recently advanced nonlinear precoding algorithms, such as the semidefinite relaxation based precoder \cite{b4}, the finite-alphabet (FA) precoder \cite{b5} and the alternating direction method of multipliers (ADMM) based precoder \cite{b6}, can improve the bit error rate (BER) performance. Generally speaking, equipping the BS with low-resolution DACs can greatly reduce the power consumption without introducing significant performance penalties by assuming perfect knowledge of CSI.

In practical situations, however, the CSI available to the BS is imperfect due to limited feedback and finite-energy
training, resulting in receivers performance degradation. Therefore, robust precoding design that takes into consideration the channel uncertainties is of great importance.

In this paper, a random matrix with finite second-order statistics is used to model the imperfection in CSI. We formulate the robust quantized precoding problem as minimization of inter user interference subject to two constraints: a discrete set (outputs of multi-bit DACs) constraint and a bounded CSI error constraint. By exploiting the overlap between the multi-bit DACs
outputs and the single-bit ones, we reformulate the original multi-variable discrete optimization problem into a binary optimization problem. Furthermore, based on the observation that the bounded CSI error constraint involves a quadratic form of some random variables, using the S-procedure lemma \cite{b7}, the bounded CSI error constraint is recast into a tractable
formulation with convex constraints. In this manner, we obtain the relaxed version of the original precoding problem, which is finally tackled by the standard semidefinite relaxation (SDR) method \cite{b8}. Lastly, the effectiveness of the proposed precoding algorithm is verified by numerical simulations.

Notations: Throughout this paper, vectors and matrices are given in
lower and uppercase boldface letters, e.g., $\bf{x}$ and $\bf{X}$, respectively. We use
${\left[ {\bf{X}} \right]_{k,l}}$ to denote the element at the $k$th row and $l$th column.
The symbols ${ \mathbb{E} \left[ {\bf{X}} \right] }$, ${ \rm{tr} \left( {\bf{X}} \right) }$, ${\mathop{\rm vec}\nolimits} \left( {\bf{X}}\right)$, ${{\bf{X}} ^{\rm T}}$, and ${{\bf{X}} ^{\rm H}}$ denote the expectation operator, the trace operator, the column-wise vectorization, the transpose,
the conjugate transpose of ${\bf{X}}$, respectively. For a vector ${\bf{x}}$, $\Re \left( {\bf{x}} \right)$, $\Im \left( {\bf{x}} \right)$,
and ${\left\| {\bf{x}} \right\|_2}$ are respectively used to represent the real part, the imaginary part and
$\ell_2$-norm of ${\bf{x}}$. The symbols $\bf{I}$ and $\bf{0}$  are respectively referred to an identity matrix and a zeros matrix with proper size.

\section{System Model}
\label{Sec:b}

We consider a single-cell MU-MIMO downlink system of one BS serving $K$ single-antenna user terminals (UTs). We assume the BS is equipped with $N$ transmit antennas and ignore the RF impairments. Besides, we follow \cite{b9} 
for the model of real-time downlink channel through a Gauss-Markov uncertainty of the form
\begin{equation}
\label{eq1}
{\bf{H}} = \sqrt {1 - \eta } {{\bf{\hat H}}} + \sqrt \eta  {\bf{E}},
\end{equation}
where ${{\bf{\hat H}}}$ is the imperfect observation of the channel available to the BS, and $\bf{E}$ denotes error matrix whose elements are independently sampled from a bounded set:
\begin{equation}
\label{eq1a}
\Gamma : = \left\{ {{\mathop{\rm tr}\nolimits} \left( {{\bf{E}}{{\bf{E}}^{\rm H}}} \right) \le {\eta}} \right\}.
\end{equation}
The parameter $\eta$ indicates the degree of uncertainty of the related channel measurement ${{\bf{\hat H}}}$. Specifically,  $\eta = 0$ means perfect CSI, the values of $0 < \eta  < 1$ correspond to partial CSI and $\eta = 1$ accounts for no CSI.

\subsection{Input-Output Relationship for MU-MIMO Downlink}

Let ${\bf{s}}$ be the $K$-dimensional UTs intended constellation points. With the knowledge of CSI $({{\bf{\hat H}}})$, the BS precodes ${\bf{s}}$ into a $N$-dimensional vector ${\bf{x}}$, satisfying the average power constraints ${\mathbb{E}_{\bf{s}}}\left[ {{\bf{x}}^\mathrm{H}}{{\bf{x}}} \right] \le P$. Assuming perfect synchronization, the received signal at UEs can be expressed as
\begin{equation}
\label{eq2}
{{\bf{y}}} = {\bf{H}}{{\bf{x}}} + {{\bf{n}}},
\end{equation}
where ${{\bf{n}}}$ is a complex vector with element $n_i$ being complex addictive Gaussian noise distributed as ${{n_i}} \thicksim \mathcal{CN} \left( {0,{\varepsilon ^2} } \right)$.

\subsection{Quantization}

Let ${\bf{P}} \in \mathbb{C}^{N \times K}$ be the precoding matrix and ${\bf{z}} = {\bf{Ps}}$ the precoded vector of the unquantized system. For the quantized MU-MIMO downlink system, each precoded signal component $z_i, \ i = 1, \cdots, N$ is quantized separately into a finite set of prescribed labels by a $B$-bit symmetric uniform quantizer $Q$. It is assumed that the real and imaginary parts of precoded signals are quantized separately. The resulting quantized signals read
\begin{equation}
\label{eq3}
{\bf{x}} = Q({\bf{z}}) = Q({\bf{Ps}}) = {Q( \Re \left( {{\bf{Ps}}} \right))} + {j Q( \Im\left( {{\bf{Ps}}} \right))}.
\end{equation}
Specially, the real-valued quantizer $Q(.)$ maps real-valued input (real parts or imaginary parts of the precoded signals) to a set of labels $\Omega = \left\{ {{l_1},{\kern 1pt} \cdots ,{l_{{2^B}}}} \right\}$, which are determined by the set of thresholds $\Gamma  = \left\{ {{\tau _1},{\kern 1pt}  \cdots ,{\tau _{{2^B} - 1}}} \right\}$, such that $ - \infty  = {\tau _1} <  \cdots  < {\tau _{{2^B-1}}} = \infty $. For a $B$-bit DAC with step size $\Delta$, the thresholds and quantization labels (outputs) are respectively given by
\begin{equation}
\label{eq4}
{\tau_b} = \Delta \left( {b - \frac{{{2^{B - 1}}}}{2}} \right), \  b = 1, \cdots ,{2^B -1},
\end{equation}
and \begin{equation}
\label{eq5}
{l_b} = \left\{ {\begin{array}{*{20}{c}}
{{\tau _b} - {\textstyle{\Delta  \over 2}},}&{b = 1, \cdots ,{2^B} - 1}\\
{\left( {{2^B} - 1} \right){\textstyle{\Delta  \over 2}},}&{b = {2^B}}
\end{array}} \right. .
\end{equation}
In the case of 1 bit DACs, the output set reduces to ${\Omega_1} = \left\{ { \pm \frac{\Delta }{2} \pm {\rm i} \frac{\Delta }{2}} \right\}$. For any output drawn from $B$-bit uniform quantifier, $x_n$ can be represented by
\begin{equation}
\label{eq6}
x_n = \sum\limits_{m=1}^M {{\omega_{m}}{v _{mn}}}, n = 1,2, \cdots, 2*N \ ,
\end{equation}
where $v _{mn} \in \Omega_1$, $M=log_2(2^B)=B$, and $ \omega _{m}$ are constant coefficients satisfying ${\omega_{m}} > 0,\sum\nolimits_{m = 1}^M {\omega_{m}^2}  = 1 .$ The relation in \eqref{eq6} indicates that the multiple DACs outputs can be represented by the linear combination of several independent single-bit DACs outputs.

\section{Robust Quantized Precoding Design}

\subsection{Problem Formulation}
Similar to \cite{b5}, we concentrate on a performance metric that minimizes inter user interference,
\[{\mathbb{E}_{\bf{s}}}\left\{ {\left\| {{\bf{s}} - \beta {\bf{Hx}}} \right\|_2^2} \right\},\]
under the average power constraint. $\beta$ is the unknown precoding factor taking into consideration the power constraint.
Specially, in the case of channel uncertainties, the quantized precoding problem can be formulated as
\begin{equation}
\label{eq7}
\begin{array}{*{20}{c}}
{\mathop {{\rm{minimize}}}\limits_{\beta  \in {R^ + }} }&{\mathbb{E}_{\bf{s}}}\left\{ {\left\| {{\bf{s}} - \beta {\bf{Hx}}} \right\|_2^2} \right\}\\
{{\rm{s.t.}}}&{{\bf{x}} \in {\Omega _B}}
\end{array},
\end{equation}
where ${\Omega _B} = \left\{ { - {l_{{2^{B - 1}}}}, \cdots , - {l_1},{l_1}, \cdots, {l_{{2^{B - 1}}}}} \right\}$. The complex-valued optimization problem in \eqref{eq7} can be equivalently rewritten as a real-valued problem:
\begin{equation}
\label{eq8}
\begin{array}{*{20}{c}}
{{\rm{minimize}}}&{\left\| {{\bf{\tilde s}} - \left( {{\bf{\tilde H}} + {\bf{\tilde E}}} \right){\bf{\tilde x}}} \right\|_2^2}\\
{{\mathop{\rm s.t.}\nolimits}}&{{\bf{\tilde x}} \in {{\tilde \Omega }_B}}
\end{array},
\end{equation}
where, with a slight abuse of notations, we define
\[\begin{array}{*{20}{l}}
{{\bf{\tilde s}} = \left[ {\begin{array}{*{20}{l}}
{\Re \left( {\bf{s}} \right)}\\
{\Im \left( {\bf{s}} \right)}
\end{array}} \right], {\bf{\tilde x}} = \beta {\left[ {\begin{array}{*{20}{c}}
	{{\Re{\left( {\bf{x}} \right)}^{\rm{T}}}}&{{\Im{\left( {\bf{x}} \right)}^{\rm{T}}}}
	\end{array}} \right]^{\rm{T}}},}\\
\begin{array}{l}
{\bf{\tilde H}} = \sqrt {1 - \eta } \left[ {\begin{array}{*{20}{c}}
{\Re \left( {{\bf{\hat H}}} \right)}&{ - \Im \left( {{\bf{\hat H}}} \right)}\\
{\Im \left( {{\bf{\hat H}}} \right)}&{\Re \left( {{\bf{\hat H}}} \right)}
\end{array}} \right],\\
{\bf{\tilde E}} = \sqrt \eta  \left[ {\begin{array}{*{20}{c}}
{\Re \left( {{\bf{E}}} \right)}&{ - \Im \left( {{\bf{E}}} \right)}\\
{\Im \left( {{\bf{E}}} \right)}&{\Re \left( {{\bf{E}}} \right)}
\end{array}} \right],\\
{{\tilde \Omega }_B} = \Re \left\{ {{\Omega _B}} \right\} \cup \Im \left\{ {{\Omega _B}} \right\}.
\end{array}
\end{array}\]

The nonlinear precoding problem in \eqref{eq8} can be tackled by the naive exhaustive search with complexity of order
$O(2^{NK})$. The unendurable computational complexities of these methods
impede their application in massive MU-MIMO.

\subsection{Robust Precoding Design}
In order to simultaneously achieve robust and efficient precoding, we reformulate the problem in \eqref{eq8} as a binary optimization problem by leveraging the favorable relation shown in \eqref{eq6}. Define an auxiliary matrix ${\bf{C}}$ and denote
\begin{equation}
\label{eqe2}
{\bf{v}} \in {{\mathbb{R}}^{2NM \times 1}} = {\left[ {{v_1}, \cdots ,{v_{2NM}}} \right]^{\rm{T}}}, \ {v_i} \in {{\tilde \Omega }_1},
\end{equation}
\begin{equation}
\label{eqe22}
{\bf{x}}_R = {\left[ {{{\begin{array}{*{20}{c}}
{{\Re({\bf{x}}^{\rm{T}} })} & \Im({\bf{x}}^{\rm{T}} )
\end{array}}}} \right]^{\rm{T}} },
\end{equation}
it follows from \eqref{eq6} that
\begin{equation}
\label{eqe3}
{\bf{x}}_R = \bf{C} \bf{v}.
\end{equation}
Specially, we have
\begin{equation}
\label{eqd1}
{\bf{C}} = \left[ {\begin{array}{*{20}{c}}
{\omega_1{{\bf{I}}_{2R}}}&{\omega_2{{\bf{I}}_{2R}}}
\end{array}} \right],
\end{equation}
and
\begin{equation}
\label{eqd2}
{\bf{C}} = \left[ {\begin{array}{*{20}{c}}
{\omega_1{{\bf{I}}_{2R}}}&{\omega_2{{\bf{I}}_{2R}}}&{\omega_3{{\bf{I}}_{2R}}}
\end{array}} \right],
\end{equation}
for the 2-bit case and the 3-bit case \footnote{For the B-bit case, we can have ${\bf{C}} = \left[ {\begin{array}{*{20}{c}}
{{\omega _1}{{\bf{I}}_{2N}}}& \cdots &{{\omega _B}{{\bf{I}}_{2N}}}
\end{array}} \right].$ Here, we only employ 1-3 bit DACs as the performance gap between MU-MIMO system with ideal DACs and the one with B-bit $(B\geq 3)$ DACs is negligible, which would be confirmed in Section \ref{smu}.}, respectively. With expressions in \eqref{eqe2} and \eqref{eqe3}, one can rewrite the precoding problem in \eqref{eq8} in the following equivalent form:
\begin{equation}
\label{eq9}
\begin{array}{*{20}{c}}
{{\rm{minimize}}}&{\left\| {{\bf{\tilde s}} - \left( {{\bf{\tilde H}} + {\bf{\tilde E}}} \right){\bf{C}}{\bf{\tilde v}}} \right\|_2^2}\\
{{\mathop{\rm s.t.}\nolimits}}&{\tilde v_1^2 =  \cdots  = \tilde v_{2NM}^2}
\end{array},
\end{equation}
where ${{\tilde v}_i} = \beta {v_i}$, for $i=1,\cdots,2NM$.

By introducing a slack variable $\varepsilon$, we can equivalently rewrite the problem in \eqref{eq9} as
\begin{subequations}
\begin{align}
{\mathop {{\mathop{\rm minmize}\nolimits} }\limits_{\bf{V}} }& \ \ \ \ \varepsilon \\
{{\mathop{\rm s.t.}\nolimits}}& \ \ \ \ {{\mathop{\rm tr}\nolimits} \left( {{\bf{TV}}} \right) \le \varepsilon } \\
{}& \ \ \ \ {{{\left[ {\bf{V}} \right]}_{1,1}} = {{\left[ {\bf{V}} \right]}_{2NM,2NM}}} \\
{}& \ \ \ \ {{{\left[ {\bf{V}} \right]}_{2NM + 1,2NM + 1}} = 1} \\
{}& \ \ \ \ {{\bf{V}} \succeq 0} \\
{}& \ \ \ \ {{\mathop{\rm rank}\nolimits} \left( {\bf{V}} \right) = 1}
\end{align}
\end{subequations}
where
\[\begin{array}{l}
{\bf{T}} = \left[ {\begin{array}{*{20}{c}}
{{{\bf{C}}^{\rm T}}{{\left( {{\bf{\tilde H}} + {\bf{\tilde E}}} \right)}^{\rm T}}\left( {{\bf{\tilde H}} + {\bf{\tilde E}}} \right){\bf{C}}}&{ - {{\bf{C}}^{\rm T}}{{\left( {{\bf{\tilde H}} + {\bf{\tilde E}}} \right)}^{\rm T}}{\bf{\tilde s}}}\\
{ - {{{\bf{\tilde s}}}^{\rm T}}\left( {{\bf{\tilde H}} + {\bf{\tilde E}}} \right){\bf{C}}}&{{{{\bf{\tilde s}}}^{\rm T}}{\bf{\tilde s}}}
\end{array}} \right],\\
{\bf{V}} = {\left[ {\begin{array}{*{20}{c}}
{{{{\bf{\tilde v}}}^{\rm T}}}&1
\end{array}} \right]^{\rm T}}\left[ {\begin{array}{*{20}{c}}
{{{{\bf{\tilde v}}}^{\rm T}}}&1
\end{array}} \right],
\end{array}\]
and ${\bf{V}} \succeq 0$ means ${\bf{V}}$ is a nonnegative definite matrix.

In the sequel, by applying the S-Procedure lemma, we rewrite the constraints
in \eqref{eq1a} and (16b) that involve quadratic inequalities in error vectors in the linear matrix inequality constraints as stated in the following.

\begin{theorem}
\label{lema1}
Given the bounded channel uncertainty in \eqref{eq1a}, the condition in (16b) can be relaxed to
\begin{subequations}
\label{eq10}
\begin{align}
{\mathop{\rm tr}\nolimits} \left( {\left( {{\bf{V}} + {\bf{W}}} \right){{{\bf{\mathord{\buildrel{\lower3pt\hbox{$\scriptscriptstyle\smile$}}
\over H} }}}^T}{\bf{\mathord{\buildrel{\lower3pt\hbox{$\scriptscriptstyle\smile$}}
\over H} }}} \right) \le \varepsilon  - 2\kappa {\eta ^2} ,\\
\left[ {\begin{array}{*{20}{c}}
{{\bf{V}} - \kappa {\bf{I}}}&{\bf{V}}\\
{\bf{V}}&{ - {\bf{W}}}
\end{array}} \right] \succeq {\bf{0}},
\end{align}
\end{subequations}
where ${\bf W} \in \mathbb{R}^{{(2NM+1)} \times {(2NM+1)}} $ is an auxiliary matrix and $\kappa$ is a slack variable.
\end{theorem}
\begin{proof}{:}
We start by restating the S-Procedure lemma.
\begin{lemma} [S-Procedure lemma, \cite{b7}] Let ${\varphi _k}\left( {\bf{o}} \right)$, for $k=1,2$, be defined as
\[{\varphi _k}\left( {\bf{o}} \right) = {{\bf{o}}^{\rm T}}{\bf{A}}_k^{\rm T}{\bf{o}} + 2{\bf{b}}_k^{\rm T}{\bf{o}} + {c_k},\]
where ${\bf{A}}_k \in \mathbb{R}^{M \times M}$, ${\bf{b}}_k \in \mathbb{R}^{M \times 1}$ and
${c_k} \in \mathbb{R}$. Suppose there exists an ${{\bf{\tilde o}}}$ with ${\varphi _2}\left( {{\bf{\tilde o}}} \right) < 0$. Then, for any ${\bf{o}}$,  the equalities
\[\left\{ \begin{array}{l}
{\varphi _1}\left( {\bf{o}} \right) \ge 0\\
{\varphi _2}\left( {\bf{o}} \right) \le 0
\end{array} \right. ,\]
hold if and only if there exists $\kappa > 0$ such that
\[\left[ {\begin{array}{*{20}{c}}
{{{\bf{A}}_1} + \kappa {{\bf{A}}_2}}&{{{\bf{b}}_1} + {{\bf{b}}_2}}\\
{{\bf{b}}_1^{\rm T} + {\bf{b}}_2^{\rm T}}&{{c_1} + \kappa {c_2}}
\end{array}} \right] \succeq {\bf 0} .\]
\end{lemma}

Since the trace operator satisfies:
\begin{equation}
\label{eqtr}
\begin{array}{l}
{\rm tr}\left( {{\bf{XY}}{{\bf{X}}^{\rm T}}} \right) = {\mathop{\rm vec}\nolimits} \left( {\bf{X}} \right)\left( {{\bf{I}} \otimes {\bf{X}}} \right){\mathop{\rm vec}\nolimits} {\left( {\bf{X}} \right)^{\rm T}}\\
{\rm tr}\left( {{\bf{X}}{{\bf{Y}}^{\rm T}}} \right) = {\mathop{\rm vec}\nolimits} \left( {\bf{X}} \right){\mathop{\rm vec}\nolimits} {\left( {\bf{Y}} \right)^{\rm T}}
\end{array},
\end{equation}
the first constraint of \eqref{eq8} can be reformulated as
\begin{equation}
\label{eqt1}
\begin{aligned}
\left\| {{\bf{\tilde s}} - \left( {{\bf{\tilde H}} + {\bf{\tilde E}}} \right){\bf{C\tilde v}}} \right\|_2^2
 = {\mathop{\rm tr}\nolimits} \left( {\left( {{\bf{\mathord{\buildrel{\lower3pt\hbox{$\scriptscriptstyle\smile$}}
\over H} }} + {\bf{\mathord{\buildrel{\lower3pt\hbox{$\scriptscriptstyle\smile$}}
\over E} }}} \right){\bf{V}}\left( {{{{\bf{\mathord{\buildrel{\lower3pt\hbox{$\scriptscriptstyle\smile$}}
\over H} }}}^T} + {{{\bf{\mathord{\buildrel{\lower3pt\hbox{$\scriptscriptstyle\smile$}}
\over E} }}}^{\rm{T}}}} \right)} \right)\\
 = {\mathop{\rm tr}\nolimits} \left( {{\bf{\mathord{\buildrel{\lower3pt\hbox{$\scriptscriptstyle\smile$}}
\over H} V}}{{{\bf{\mathord{\buildrel{\lower3pt\hbox{$\scriptscriptstyle\smile$}}
\over H} }}}^T}} \right) + {\mathop{\rm tr}\nolimits} \left( {{\bf{\mathord{\buildrel{\lower3pt\hbox{$\scriptscriptstyle\smile$}}
\over E} V}}{{{\bf{\mathord{\buildrel{\lower3pt\hbox{$\scriptscriptstyle\smile$}}
\over E} }}}^{\rm{T}}}} \right) + 2{\mathop{\rm tr}\nolimits} \left( {{\bf{\mathord{\buildrel{\lower3pt\hbox{$\scriptscriptstyle\smile$}}
\over H} V}}{{{\bf{\mathord{\buildrel{\lower3pt\hbox{$\scriptscriptstyle\smile$}}
\over E} }}}^{\rm{T}}}} \right)\\
 = {\mathop{\rm tr}\nolimits} \left( {{\bf{\mathord{\buildrel{\lower3pt\hbox{$\scriptscriptstyle\smile$}}
\over H} V}}{{{\bf{\mathord{\buildrel{\lower3pt\hbox{$\scriptscriptstyle\smile$}}
\over H} }}}^T}} \right) + {{{\bf{\mathord{\buildrel{\lower3pt\hbox{$\scriptscriptstyle\smile$}}
\over e} }}}^{\rm{T}}}\left( {{\bf{I}} \otimes {\bf{V}}} \right){\bf{\mathord{\buildrel{\lower3pt\hbox{$\scriptscriptstyle\smile$}}
\over e} }} + {\mathop{\rm vec}\nolimits} {\left( {{\bf{\mathord{\buildrel{\lower3pt\hbox{$\scriptscriptstyle\smile$}}
\over H} V}}} \right)^T}{\bf{\mathord{\buildrel{\lower3pt\hbox{$\scriptscriptstyle\smile$}}
\over e} }} \le \varepsilon,
\end{aligned}
\end{equation}
where ${\bf{\mathord{\buildrel{\lower3pt\hbox{$\scriptscriptstyle\smile$}}
\over H} }} = \left[ {\begin{array}{*{20}{c}}
{{\bf{\tilde HC}}}&{{\bf{\tilde s}}}
\end{array}} \right],{\bf{\mathord{\buildrel{\lower3pt\hbox{$\scriptscriptstyle\smile$}}
\over E} }} = \left[ {\begin{array}{*{20}{c}}
{{\bf{\tilde H\tilde E}}}&{\bf{0}}
\end{array}} \right],{\bf{\mathord{\buildrel{\lower3pt\hbox{$\scriptscriptstyle\smile$}}
\over e} }} = {\mathop{\rm vec}\nolimits} \left( {{\bf{\mathord{\buildrel{\lower3pt\hbox{$\scriptscriptstyle\smile$}}
\over E} }}} \right)$, and $\otimes$ denotes the Kronecker product.
Besides, the complex-valued constraint in \eqref{eq1a} can be equivalently expressed by a real-valued one, of form
\begin{equation}
\label{eqt2}
\begin{array}{l}
\left\| {{\bf{\mathord{\buildrel{\lower3pt\hbox{$\scriptscriptstyle\smile$}}
\over e} }}} \right\|_2^2 = {\rm tr}\left( {{\bf{\tilde E}}{{{\bf{\tilde E}}}^{\rm T}}} \right)\\
 = 2\eta \ {\rm tr}\left( {\Re \left( {\bf{E}} \right)\Re {{\left( {\bf{E}} \right)}^{\rm T}} + \Im \left( {\bf{E}} \right)\Im {{\left( {\bf{E}} \right)}^{\rm T}}} \right)\\
 = 2{\rm tr}\left( {{\bf{E}}{{\bf{E}}^{\rm H}}} \right) \le 2{\eta ^2}
\end{array}.
\end{equation}

According to Lemma 1, the inequalities in \eqref{eqt1} and \eqref{eqt2} hold if and only if there exists $\kappa >0$ such that
\begin{equation}
\label{eqt3}
\left[ {\begin{array}{*{20}{c}}
{ - {\bf{I}} \otimes {\bf{V}} + \kappa {\bf{I}}}&{ - {\mathop{\rm vec}\nolimits} \left( {{\bf{\mathord{\buildrel{\lower3pt\hbox{$\scriptscriptstyle\smile$}}
\over H} V}}} \right)}\\
{ - {\mathop{\rm vec}\nolimits} {{\left( {{\bf{\mathord{\buildrel{\lower3pt\hbox{$\scriptscriptstyle\smile$}}
\over H} V}}} \right)}^{\rm T}}}&{ - {\mathop{\rm tr}\nolimits} \left( {{\bf{\mathord{\buildrel{\lower3pt\hbox{$\scriptscriptstyle\smile$}}
\over H} V}}{{{\bf{\mathord{\buildrel{\lower3pt\hbox{$\scriptscriptstyle\smile$}}
\over H} }}}^{\rm T}}} \right) + \varepsilon  - 2\kappa {\eta ^2}}
\end{array}} \right]  \succeq 0.
\end{equation}
Using Schur's complement \cite{b11}, \eqref{eqt3} can be reformulated as
\begin{equation}
\label{eqt4}
\begin{aligned}
- \ {\mathop{\rm vec}\nolimits} {\left( {{\bf{\mathord{\buildrel{\lower3pt\hbox{$\scriptscriptstyle\smile$}}
\over H} V}}} \right)^{\rm T}}{\left( { - {\bf{I}} \otimes {\bf{V}} + \kappa {\bf{I}}} \right)^{ - 1}}{\mathop{\rm vec}\nolimits} \left( {{\bf{\mathord{\buildrel{\lower3pt\hbox{$\scriptscriptstyle\smile$}}
\over H} V}}} \right) \ \ \ \ \\
- \ {\mathop{\rm tr}\nolimits} \left( {{\bf{\mathord{\buildrel{\lower3pt\hbox{$\scriptscriptstyle\smile$}}
\over H} V}}{{{\bf{\mathord{\buildrel{\lower3pt\hbox{$\scriptscriptstyle\smile$}}
\over H} }}}^{\rm T}}} \right) + \varepsilon  - 2\kappa {\eta ^2} \ge 0,
\end{aligned}
\end{equation}
or equivalently,
\begin{equation}
\label{eqt5}
\begin{aligned}
- \ {\mathop{\rm tr}\nolimits} \left( {\left( {{\bf{V}}{{\left( { - {\bf{V}} + \kappa {\bf{I}}} \right)}^{ - 1}}{\bf{V}} - {\bf{W}}} \right){{{\bf{\mathord{\buildrel{\lower3pt\hbox{$\scriptscriptstyle\smile$}}
\over H} }}}^{\rm T}}{\bf{\mathord{\buildrel{\lower3pt\hbox{$\scriptscriptstyle\smile$}}
\over H} }}} \right) \ \ \ \ \ \ \ \\ - \ {\mathop{\rm tr}\nolimits} \left( {\left( {{\bf{V}} + {\bf{W}}} \right){{{\bf{\mathord{\buildrel{\lower3pt\hbox{$\scriptscriptstyle\smile$}}
\over H} }}}^{\rm T}}{\bf{\mathord{\buildrel{\lower3pt\hbox{$\scriptscriptstyle\smile$}}
\over H} }}} \right) + \varepsilon  - 2\kappa {\eta ^2}  \ge 0,
\end{aligned}
\end{equation}
where ${\bf W} \in \mathbb{R}^{{(2NM+1)} \times {(2NM+1)}} $ is an auxiliary matrix. Finally, using Lemma 2 in \cite{b12}, we obtain \eqref{eq10}, which completes the proof.
\end{proof}

With Theorem \ref{lema1}, the optimization problem in \eqref{eq9} can be reformulated as
\begin{equation}
\label{eq11}
\begin{array}{*{20}{c}}
{\mathop {{\mathop{\rm minmize}\nolimits} }\limits_{{\bf{V}},{\bf{W}}, {\kappa > 0}} }&\varepsilon \\
{{\mathop{\rm s.t.}\nolimits}}&{ \rm (16b)-(16f), \ (17a), \ (17b)}
\end{array}.
\end{equation}
The precoding problem in \eqref{eq11} is still NP-hard due to the non-convex rank constraint in (16f). We follow the standard strategy \cite{b9} that relaxes the problem in \eqref{eq11} by omitting the non-convex constraint (16f). Then the remaining convex problem in \eqref{eq11} can be efficiently solved by numerical algorithms, e.g., the interior point method \cite{b13}.  Finally, let ${\hat {\bf{v}}}$ denote the solution of \eqref{eq11}, using the rounding strategy as described in \cite{b8},
we can obtain the desired real-valued precoded vector by
\begin{equation}
\label{eqp10}
\hat{\bf{ x}}_R = {\bf{C}} \ {\mathop{\rm sign}\nolimits} \left[ {\hat{\bf{v}}} \right].
\end{equation}
Then, the complex-valued precoded vector can be determined as
\begin{equation}
\label{eqp11}
{\hat {\bf{x}}} = {\left[ {\hat{\bf{ z}}_R} \right]_{1:N}} + {\rm i}  {\left[ {\hat{\bf{ z}}}_R \right]_{N + 1:2N}}.
\end{equation}
We refer to the proposed method as RSDR. The complexity analysis of RSDR is presented in the following.

\subsection{Complexity Analysis}

The quantized precoding problem \eqref{eq7} can be directly tackled by the naive exhaustive
search (NES) \cite{b18} with complexity of order $O(4^N)$ or by the sphere decoding (SD) method \cite{b10, b19} of exponential complexity. In this paper, the original quantized precoding problem is recast into a tractable formulation \eqref{eq11} and finally solved by the interior point method.

Since the proposed RSDR
has a similar form to the standard SDP in \cite{b20}, we can borrow the analytical result therein. According to \cite{b20}, the worst-case complexity of solving the SDP is $O\left( {\sqrt {N + K} \left( {{N^3} + {N^2}K + {K^3}} \right)\log \left( {1/\varepsilon } \right)} \right)$, where $K$ is the number of constraints in \eqref{eq11} and $\varepsilon$ is the solution precision. In practice, it has been shown in \cite{b20} that $K$ scales more typically with $N$, and hence the number of operations is bounded by $O\left(N^3 \right)$.

\section{Simulations}
\label{smu}
\subsection{System setup}

This section evaluates the performance of the proposed precoder via numerical simulations in comparison with the ADMM precoder and the FA precoder. Each provided result is an average over 1000 independent runs.

We use the Lloyd quantization method \cite{b15} to determine the quantization outputs \eqref{eq5} and the least square regression method \cite{b16} to calculate the coefficients in \eqref{eqd1} and \eqref{eqd2}. Specially, for the 2-bit case, we have ${\omega _1} = 1/\sqrt 5, \ {\omega _2} = 2/\sqrt 5$. For the 3-bit case, we have ${\omega _1} = 1/\sqrt 21, \ {\omega _2} = 2/\sqrt 21, $ and ${\omega _3} = 4/\sqrt 21$. We assume that the entries of ${\bf{E}}$ are zero-mean
Gaussian distributed, i.e., ${\left[ {{\bf{E}}} \right]_{i,j}} \sim \mathcal{CN}\left( {0,1} \right)$.
The Water Filling (WF) precoder \cite{b14} with infinite-resolution DACs and perfect CSI is regarded as the benchmark.

\subsection{Performance Evaluation}

\begin{figure}[htbp]
\centerline{\includegraphics[width=0.475\textwidth]{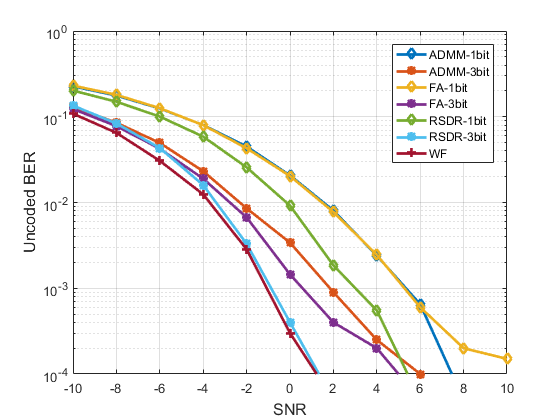}}
\caption{Uncoded BER of the compared precoders with different quantization levels and QPSK signaling.}
\label{fig1}
\end{figure}

In Fig. \ref{fig1}, we consider a 128-antenna MU-MIMO system with 10 UTs. The CSI-related parameter is $\eta = 0.2$. It can be observed that the performance of all the compared precoders, in terms of BER, improve as a result of a increased bits of DACs.

One can also observe from Fig. \ref{fig1} that
the proposed RSDR precoder outperforms the ADMM precoder and the FA precoder.  For instance, with same average transmit power and 3-bit DACs, the proposed precoding scheme can attain 3 dB of target BER whenever ${SNR} > -1dB$, whereas the FA precoder and the ADMM precoder respectively require higher SNR conditions (${SNR} > 1dB$ and ${SNR} > 2dB$), to support 3 dB of BER target at the same CSI condition. Besides, the performance gap between the ideal ZF precoder (infinite DACs) with perfect CSI and the proposed precoder with 3-bit DACs is negligible, indicating the robustness of the the proposed precoding scheme.

\begin{figure}[htbp]
\centerline{\includegraphics[width=0.475\textwidth]{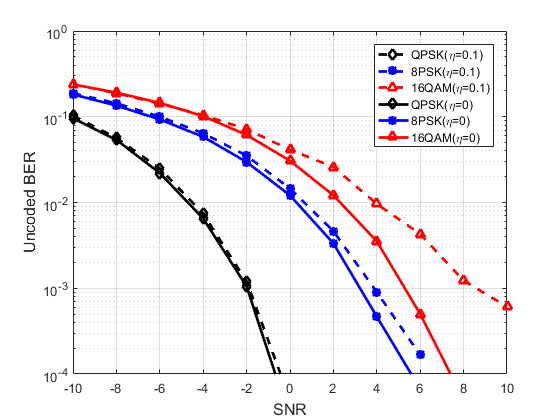}}
\caption{Uncoded BER of the proposed precoder in the case of a MU-MIMO system with 3bit DACs and different modulation schemes.}
\label{fig2}
\end{figure}

In Fig. \ref{fig2}, we investigate the performance of the proposed precoder for different modulation schemes, QPSK, 8PSK and 16QAM. The BS is equipped with 128 antennas and serves 16 UTs with different channel uncertainties, e.g., $\eta = 0, 0.1$. It can be seen that the proposed RSDR precoder shows robustness for all the considered modulation schemes under channel uncertainties. Specially, compared to the ideal case ($\eta = 0$), the performance gaps are about 0.2 dB, 1 dB, and 3 dB, for QPSK, 8PSK and 16QAM with a target BER of $10^{-3}$, respectively.

\begin{figure}[htbp]
	\centerline{\includegraphics[width=0.475\textwidth]{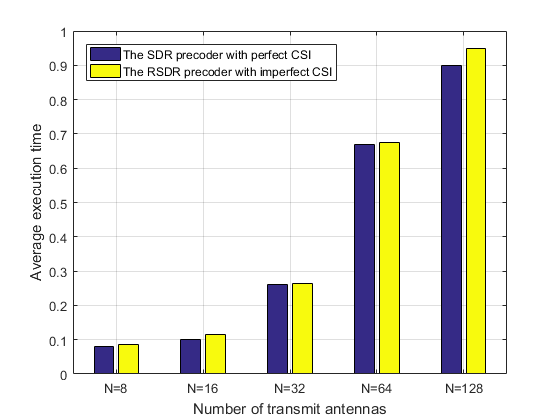}}
	\caption{Runtime comparison of the proposed precoder versus the number of the BS antennas. We consider a N-antenna MU-MIMO system using QPSK modulation. The number of UTs is U = 10.}
	\label{fig3}
\end{figure}

Fig. \ref{fig3} compares the average execution time (in seconds) of RSDR in the case of different channel conditions. Compared to the SDR precoder \cite{b4} that
has perfect access to CSI, the proposed RSDR precoder only incurs two added constraints. The algorithms are performed on a desktop PC with an Intel Core I7-7700K CPU at 4.2 GHz with 32 GB RAM.   It can be seen that the average execution time of RSDR increases as the number of the BS antennas increases. Compared to the perfect CSI case, the RSDR can obtain robust performance at the expense of slightly more numerical computations,  indicating the effectiveness of the new precoder.

In summary, the proposed RSDR enables robust precoding against channel uncertainty in polynomial time \cite{b20}. We note here that the efficiency of the proposed precoder needs to be improved to meet the demand in 5G or 6G communication systems, in which the BS is foreseen to be equipped with hundreds, even thousands of antennas.

\section{Conclusions}
\label{Sec:f}

In this paper, we have investigated the quantized precoding problem for the MU-MIMO in the presence of imperfect CSI. By exploiting the bridging relation between the multi-bit DACs outputs and the single-bit ones, and taking advantaging of the S-procedure lemma, the original precoding problem is reformulated into a tractable formulation and solved by the SDR method. Simulation results have verified that the proposed precoder is robust to various channel uncertainties and can support the MU-MIMO system with higher-order modulations.


\section{Acknowledgements}
\label{Sec:g}

The authors would like to thank for technical experts: Tian Pan, Guangyi Yang, Yingzhe Li,and Rui Gong, all from Huawei Technologies for their fruitful discussions.

\vspace{12pt}

\end{document}